\documentclass{llncs}

\usepackage{amssymb,amsmath,amsfonts,graphicx,psfig}
\usepackage{times}
\usepackage{xspace}
\usepackage{subfigure}  
\usepackage{xcolor}
\usepackage{algorithm}
\usepackage[noend]{algorithmic}
\usepackage{hyperref}

\usepackage{color}
\usepackage{xcolor}

\newcommand{\algo}[1]{#1}

\newcommand{\Deltaa}{\ensuremath{\scalebox{.85}{$\Delta$}}}

\newcommand{\prob}{server renting problem\xspace}
\newcommand{\OPT}{\ensuremath{\operatorname{\textsc{Opt}}}\xspace}
\newcommand{\NF}{\ensuremath{\operatorname{\textsc{Nf}}}\xspace}
\newcommand{\MTF}{\ensuremath{\operatorname{\textsc{Mtf}}}\xspace}
\newcommand{\opt}{\ensuremath{\operatorname{\textsc{Opt}}}\xspace}

\newcommand{\longv}[1]{\hspace{-.11cm}}

\newenvironment{myquote}{\list{}{\leftmargin=0.15in \rightmargin=0.15in}\item[]}{\endlist}

\begin{document}

\title{Efficient Online Strategies for Renting \\ Servers in the Cloud}

\author{Shahin Kamali,
Alejandro L\'opez-Ortiz}
\date{}

\pagestyle{plain}

\institute{
University of Waterloo, Canada.
}
\maketitle

\begin{abstract}
In Cloud systems, we often deal with jobs that arrive and depart in an online manner. Upon its arrival, a job should be assigned to a server. Each job has a size which defines the amount of resources that it needs. Servers have uniform capacity and, at all times, the total size of jobs assigned to a server should not exceed the capacity. This setting is closely related to the classic bin packing problem. The difference is that, in bin packing, the objective is to minimize the total number of used servers. In the Cloud, however, the charge for each server is  proportional to the length of the time interval it is rented for, and the goal is to minimize the cost involved in renting all used servers. Recently, certain bin packing strategies were considered for renting servers in the Cloud [Li et al. SPAA'14]. There, it is proved that all Any-Fit bin packing strategy has a competitive ratio of at least $\mu$, where $\mu$ is the max/min interval length ratio of jobs. It is also shown that \algo{First Fit} has a competitive ratio of $2\mu + 13$ while \algo{Best Fit} is not competitive at all. We observe that the lower bound of $\mu$ extends to all online algorithms. We also prove that, surprisingly, \algo{Next Fit} algorithm has competitive ratio of at most $2 \mu +1$. We also show that a variant of \algo{Next Fit} achieves a competitive ratio of $K \times max\{1,\mu/(K-1)\}+1$, where $K$ is a parameter of the algorithm. In particular, if the value of $\mu$ is known, the algorithm has a competitive ratio of $\mu+2$; this improves upon the existing upper bound of $\mu+8$. Finally, we introduce a simple algorithm called \algo{Move To Front} (\MTF) which has a competitive ratio of at most $6\mu + 7$ and also promising average-case performance. We experimentally study the average-case performance of different algorithms and observe that the typical behaviour of \MTF is distinctively better than other algorithms.
\end{abstract}

\section{Introduction}
Bin packing is a classic problem in the context of online computation. The input is a sequence of \emph{items} of different \emph{sizes} which appear in a sequential, online manner. The goal is to \emph{place} these items into a minimum number of \emph{bins} of uniform capacity so that the total size of items in each bin is no more than the uniform capacity of the bins. It is often assumed that bins have size 1 and items have a positive size no more than 1. The problem is online in the sense that, upon receiving an item, an algorithm should place it into a bin without any knowledge about the (size of) incoming items. A simple online strategy is \algo{Next Fit} (\NF) in which there is a single \emph{open} bin at each time. If an incoming item fits in the open bin, the algorithm places it there; otherwise, it \emph{closes} the open bin and opens a new bin for the item. Clearly, if we want to minimize the number of bins, 
there is no benefit in closing a bin. \algo{First Fit} is an online algorithm that never closes a bin and places an incoming item in the first bin that has enough space for the item; if such a bin does not exist, it opens a new bin. In First Fit, the bins are maintained in the order that they are opened. \algo{Best Fit} works similarly to \algo{First Fit} except that it maintains bins in the decreasing order of their fill \emph{level}. The level of a bin is the total size of items placed in the bin. 
Note that \algo{First Fit} and \algo{Best Fit} are greedy algorithm in the sense that they avoid opening new bins unless they have to. The algorithms with this property are called Any Fit algorithm. 

In many cloud systems, a set of \emph{jobs} appear in an online manner that should be assigned to servers. Each job has a \emph{load} which defines the amount of resources that it needs. Depending on the application, the load of a job might be defined through its memory requirement, GPU resource usage, bandwidth usage, or a function of all of them.
In cloud gaming systems, different instances of computer games are created in an online fashion and run in cloud servers while players interact with the servers via thin clients \cite{Huang13,LiTang14}. Here, an instance of a game is a job which, depending on the game and the number of users involved in it, has a load.  In case of computer games, the load of a job is mainly defined through the amount of GPU resources that it demands \cite{LiTang14}. 

With the above definition, any online bin packing algorithm can be used to assign jobs to servers. A job of load $x$ can be treated as an item of size $x$ which is assigned to a server (bin) of certain capacity. In this paper, we interchangeably use the terms `job' and `item' as well as `server' and `bin'. There are, however, distinctions between assigning jobs to servers and the bin packing problem. First, jobs depart the system after they complete; however, the classic bin packing is static in the sense that items are assumed to remain in the bins. When a job arrives, it is not clear when it completes and an algorithm should place it without any knowledge about its departure time.
A more important difference between the two problems is that, in the bin packing problem, the objective is to minimize the number of used bins. In other words, we can think of bins as servers that one can \emph{buy} and we would like to minimize the cost by buying a smaller number of servers. In the cloud, however, we want to \emph{rent} the servers from cloud service providers. For example, gaming companies such as OnLive \cite{onlive} and GaiKai \cite{Gaikai} offer cloud gaming services which are execute in public clouds like Amazon EC2 \cite{AmazonEC2}. A rented server is charged by its usage (often in hourly or monthly basis). So, in order to minimize the cost, we need to minimize the total time that servers are rented. In doing so, an algorithm \emph{releases} a server when all the assigned jobs to it are complete. 

\begin{definition} 
In the \emph{\prob}, a set of jobs (items) appear in an online manner. Each job has a load (size) that defines the amount of resources that it needs. Upon its arrival, a job should be assigned to a server (bin). Servers have uniform capacity and the total load of jobs assigned to a server should not exceed the capacity. Besides the arrival time, each job has a departure time that indicates when it leaves the system. The length of the interval between the arrival and departure time of a job is referred to as the \emph{length} of the job. Upon the arrival of a job, its length 
is unknown to the online algorithm. To assign a job to a server, an online algorithm might open (rent) a new server or place it to any of the previously opened servers. When all jobs assigned to a server depart, that server is released. The goal is to minimize the total usage time of servers . More precisely, assuming that an algorithm opens $m$ bins $B_1, \ldots, B_m$, the total cost of the algorithm is $\sum_{i=1}^m t_i$, where $t_i$ is the length of time period that $B_i$ has been open for. Without loss of generality, we assume the capacity of servers to be 1 and jobs have size at most 1. Also, we assume the length of jobs to be at least $\Deltaa$ and at most $\mu \Deltaa$ where $\mu \geq 1$.
\end{definition}

When studying the \prob, we are mostly interested algorithms which have good worst-case and average-case performance.  
For measuring the worst-case performance, we compare an online algorithm with an optimal offline algorithm \OPT that knows the entire sequence (all arrival times, lengths and sizes) in advance. An algorithm is said to be $c$-competitive (more precisely, \emph{asymptotic} $c$-competitive)
if the cost of serving any input sequence never exceeds $c$ times the cost of \OPT within an additive constant. 

\subsection{Previous Work and Contribution}
The Bin Packing problem has been widely studied over the past few decades. It is known that \algo{Next Fit} is 2-competitive while \algo{Best Fit} and \algo{First Fit} are both 1.7-competitive \cite{JoDUGG74}. Generally, any Any Fit algorithm that avoids placing items in the bin with the lowest level is 1.7-competitive (these algorithms are called Almost Any Fit). 
The \algo{Harmonic} family of algorithms is another class of bin packing algorithms which are based on placing items of similar sizes together in the same bins. These algorithm generally have better competitive ratios than Any Fit algorithms. The best member of this family is \algo{Harmonic++} with a competitive ratio of 1.5888 \cite{Seid02}. However, these algorithms are rarely used in practice since of their poor average-case performance.
It is known that no online algorithm can be better than 1.54037-competitive \cite{BalBek12}.

Coffman et al. \cite{CoGaJo83} studied a dynamic version of the bin packing problem in which items arrive and depart the system. In that variant, it is assumed that the length of any item is revealed upon its arrival. It is proved that the competitive ratio of \algo{First Fit} is between $2.75$ and $2.89$ while no online algorithm can do better than $2.5$ \cite{CoGaJo83}. For the discrete version of the problem, where each item has size $1/k$ for some integer $k$, the competitive ratio of Any Fit algorithms is 3 while no online algorithm can do better than 2.48 \cite{ChLaWo08}. Note that in all these results, the objective is to minimize the number of opened bins.

The online \prob as defined above was recently introduced by Le et al. \cite{LiTang14} (some terms and notations in this paper are also borrowed from \cite{LiTang14}). There, the authors prove that no Any Fit algorithm can be better than $\mu$ competitive. Recall that $\mu$ is the ratio between the length of the largest and the smallest item in the sequence. They also proved that the competitive ratio of \algo{First Fit} is $\frac{k}{k-1} \mu + \frac{6k}{k-1} + 1$ when the size of items are upper bounded by $1/k$ ($k$ can be any positive value). In particular, when $k=1$ (when there is no restriction on item sizes), \algo{First Fit} is $2\mu+13$-competitive. On the other hand, they prove that \algo{Best Fit} is not competitive, i.e., it does not achieve a constant competitive ratio. This result is somewhat surprising as \algo{Best Fit} is usually considered to be the superior algorithm for many bin packing applications. In \cite{LiTang14} a slight modification of the \algo{First Fit} algorithm is introduced which achieves a competitive ratio of $\frac{8}{7} \mu + 55/7$ when the value of $\mu$ is not known to the algorithm and a competitive ratio of $\mu+8$ when the value of $\mu$ is known.

In this paper, we study the \prob and improve some of the results presented in \cite{LiTang14}. We first observe that the lower bound of $\mu$ presented for competitiveness of any Any Fit algorithm can be extended to any online algorithm. The focus of \cite{LiTang14} has been on studying Any Fit algorithms. In the standard bin packing, Any Fit algorithms have an advantage over bounded-space algorithms which close the bins (since there is no harm in keeping bins open). However, we show that this is not the case for the \prob and it makes sense to close some servers to avoid placing new items in the bins opened by older items. In particular, we show that the competitive ratio of the \algo{Next Fit} algorithm is at most $\frac{k}{k-1} \mu + 1$ when items are smaller than $1/k$ and $2\mu + 1$ in the general case. Note that this is much better than the ratio  $2 \mu + 13$ of \algo{First Fit}. We also introduce variants of \algo{Next Fit} which achieve a ratio of $K \times \max \{ 1, \frac{\mu}{K-1}\} + 1$, where $K$ is a parameter of the algorithm and can be any positive value. In particular, if the value of $\mu$ is known, we get an algorithm with competitive ratio of $\mu+2$ which is better than $\mu+8$ of the algorithm presented in \cite{LiTang14}. 

Although \algo{Next Fit} has a superior competitive ratio compared to \algo{Best Fit} and \algo{First Fit}, unfortunately, it has a poor average-case performance relative to these algorithms. Our experiments indicate that, for sequences generated uniformly at random, \algo{Best Fit} performs generally better than the other two algorithms. To address this issue, we introduce a simple Any Fit algorithm called \algo{Move To Front} (\MTF) which outperforms \algo{Best Fit} and other algorithms on random sequences. Moreover, in contrast to \algo{Best Fit}, \MTF is competitive and has a competitive ratio of at most $6 \mu + 7$.


\section{Preliminaries}
In this section, we present some basic results about the \prob. First, we show that any online algorithm for the \prob has a competitive ratio of at least $\mu$. Our lower bound sequence is similar to that of \cite{LiTang14} and is composed of items with uniform sizes and different lengths.

\begin{theorem}\label{lower}
The competitive ratio of any online algorithm for the \prob is at least $\frac{\mu}{1+ \epsilon (\mu-1)}$ where $\epsilon$ is a lower bound for the size of items. 
\end{theorem}

\begin{proof}
Recall that the lengths of all items are at least $\Deltaa$ and at most $\mu \Deltaa$. Consider a sequence which is defined through phases. Each phase starts with $\frac{1}{\epsilon^2}$ items of size $\epsilon$. To place these items, any algorithm has to open at least $1/\epsilon$ bins. At time $\Deltaa$, $\frac{1}{\epsilon^2}-\frac{1}{\epsilon}$ items depart in an adversarial manner so that there is a single of item of size $\epsilon$ in $1/\epsilon$ bins (some bins might get released at this time). The remaining items stay for a period of length $\mu \Deltaa$ and the online algorithm keeps a single bin for each of them. At time $\mu \Deltaa$, all items depart and the phase ends. The cost of the online algorithm for each phase is at least $\mu  \Deltaa / \epsilon$ (it keeps $1/\epsilon$ bins for a period of $\mu \Deltaa$). \OPT places items which have length $\mu  \Deltaa$ together in a single bin and incurs a cost of $\mu  \Deltaa$ for them. Other $\frac{1}{\epsilon^2}-\frac{1}{\epsilon}$ items are placed tightly together in $1/\epsilon - 1$ bins for a period of length $ \Deltaa$ (after which all they leave and their bins get closed. The cost of \OPT for these items will be $ \Deltaa/\epsilon -  \Deltaa$. In total, the cost of \OPT will be $\mu  \Deltaa +  \Deltaa/\epsilon - \Deltaa$ and the competitive ratio of the algorithm will be $ \frac{\mu \Deltaa/\epsilon}{\mu  \Deltaa +  \Deltaa/\epsilon -  \Deltaa} = \frac{\mu}{1+ \epsilon (\mu-1)}$. \qed
\end{proof}

Next, we introduce two lower bounds for the cost of $\OPT$ for serving any input sequence. We say an item $x$ is \emph{active} at time $t$ if $t$ lies in the interval between the arrival and the departure time of $a$. Let the \emph{span} of an input sequence $\sigma$ denote the length of time in which at least one item in $\sigma$ is active. Clearly, the cost of any algorithm (including the offline \OPT) for serving $\sigma$ is at least equal to the span of $\sigma$. 
Define the \emph{resource utilization} of an item as the product of its size and its length. This way, the cost of any algorithm for $\sigma$ is at least equal to $util(\sigma)$, that is, the total resource utilization of items in $\sigma$. 
So, the cost of an optimal algorithm for an input sequence is always lower bounded by the span of the sequence and also by the total utilization of the sequence.

\begin{proposition}\label{prp1}
For any input sequence $\sigma$, the cost of an optimal offline algorithm $\OPT$ is at least equal to $span(\sigma)$ and $util(\sigma)$, namely, the span of $\sigma$ and also the total resource utilization of items in $\sigma$.
\end{proposition}

When we allow arbitrary small items, Theorem \ref{lower} indicates that all algorithms have a competitive ratio of at least $\mu$. This suggests that when item sizes are larger than a fixed value, better competitive ratios can be achieved. Consider a sequence $\sigma$ in which all item sizes are larger than $1/k$ for some positive value $k$. The cost of any algorithm is at most equal to the total length of all items denoted by $L(\sigma)$ (which happens when no two items share a bin). On the other hand, 
the total resource utilization of items, and consequently cost of \opt, is at least $L(\sigma)/k$. So, we get the following.

\begin{proposition}\cite{LiTang14}\label{pipipi2}
When items sizes are lower bounded by $1/k$ ($k$ is a positive value), the competitive ratio of any online algorithm for the \prob is at most equal to $k$.
\end{proposition}

\section{\algo{Next Fit} Algorithm}

In this section, we analyze the \algo{Next Fit} algorithm for the \prob. Recall that for the bin packing problem \algo{Next Fit} keeps exactly one bin open at any given time. If an incoming item does not fit in the open bin, it closes the bin and opens a new bin. For the \prob, we distinguish between \emph{closing} and \emph{releasing} a bin. 
When an item does not fit in the open bin, the algorithm closes the bin and does not refer to it. Such a bin remains in the system (i.e., is being rented) until all items which are placed there depart and it becomes released. 

\begin{example}
Consider sequence $\left\langle a=(0.3,1,5), b=(0.4,2,6), c=(0.4,3,5), \ldots  \right\rangle$ of items. The first element of each tuple indicates the size and the second and third respectively indicate arrival and departure times of an item. At time 1, item $a$ arrives and is placed in the single open bin. At time 2, item $b$ arrives and is placed in the same bin (the level of the bin will be 0.7). At time 3, item $c$ arrives which does not fit in the open bin; hence, the current open bin is closed and a new bin is opened for $c$. The closed bin remains in the system (and a rental cost is paid for it) until time 6 where item $b$ departs and the bin gets released.
\end{example}

\begin{theorem}\label{NFTH}
The competitive ratio of \algo{Next Fit} for the \prob is at most $\frac{\mu}{1-1/k} + 1$ for serving sequences in which item sizes are no more than $1/k$. If $k<2$, the ratio is at most $2\mu + 1$.
\end{theorem}

\begin{proof}
Consider an arbitrary sequence $\sigma$ and assume \algo{Next Fit} opens $m$ bins $B_1, \ldots, B_m$ for serving $\sigma$. Let $st_i$ denote the length of the time interval at which the server $B_i$ has been rented, i.e., the time after it is opened and before it is released. We refer to this period as \emph{stretch} of $B_i$. Let $\NF(\sigma)$ denote the cost of \algo{Next Fit} for serving $\sigma$; we have $\NF(\sigma) = \sum_1^m st_i$. The stretch of $B_i$ can be partitioned into two period. First, the interval between its opening time and when \algo{Next Fit} closes the bin. 
The second period is the time between the bin gets closed and when it gets released. Let $st^1_i$ and $st^2_i$ denote the lengths of first and second period of $B_i$ ($st^1_i + st^2_i = st_i$). If a bin gets closed before being released, the second period will be empty ($st^2_i = 0)$ (see Figure \ref{fig:intervals}).

Let $p \leq m$ denote the number of bins which are closed before being revealed. We call these bins \emph{critical} bins and for them we have $st^2_i \neq 0$.
The main observation is that, when a bin gets closed, it takes a time of length at most $\mu$ before it gets released, i.e., the second period of each bin has a length of at most $\mu$ ($st^2_i \leq \mu$ for all $i$). This is because no new item is placed in the bin in the second period. So, the total rental time for the second period of all bins is no more than $p \times \mu$. 
On the other hand, the total rental time of the first periods of all bins is no more than the span of input sequence. This is because the first period of a bin starts when that of previous bin is finished (no two bins are in their first stretch period at the same time). So we have 

\begin{equation}
\NF(\sigma) = \sum\limits_{i=1}^m st_i = \sum\limits_{i=1}^m st^1_i + \sum\limits_{i=1}^p st^2_i \leq span(\sigma) + p \times \mu \Deltaa 
\label{equa1}
\end{equation}

Assume that all items in $\sigma$ are smaller or equal to $1/k$ for $k\geq 2$. At the time of being closed, all critical bins have a level of at least $1-1/k$ (otherwise the item that caused opening of a new bin could fit in such a bin). This implies that the number of critical bins (i.e., $p$) cannot be more than $\frac{\omega(\sigma)}{1-1/k}$ 
where $\omega(\sigma)$ is the total size of items in $\sigma$. Let $util(\sigma)$ denote the total resource utilization of items in $\sigma$. Since the length of each item is at least $\Deltaa$, we have $util(\sigma) \geq \omega(\sigma) \times \Deltaa$, and by Proposition \ref{prp1} $\omega(\sigma) \leq \OPT(\sigma)/\Deltaa$. Consequently, $p \leq \frac{\opt(\sigma)}{(1-1/k)  \Deltaa}$. Also, by Proposition \ref{prp1}, $span(\sigma) \leq \OPT(\sigma) $. Plugging these into Equation \ref{equa1}, we get the following inequality which completes the proof. 
$$\NF(\sigma) \leq \OPT(\sigma) + \frac{\OPT(\sigma)}{1-1/k} \times \mu  $$

Next, assume $k \leq 2$. We define the \emph{amortized level} of a critical bin $B$ as the sum of the size of the item that closes $B$ and the total size of items in $B$ (at the time that Next Fit closes it). By definition of \NF, the amortized level of all critical bins is more than 1. At the same time, the size of each item is added at most twice in the total amortized cost (once as a part of a critical bin and once as the time that closes a critical bin). Hence, the total sum of the amortized levels of all critical bins is at most twice the total size of sequence. This implies that the number of critical bins is no more than twice the total size of items in $\sigma$, i.e., $p \leq 2 \omega(\sigma) \leq 2 \opt(\sigma)/\Deltaa$. Applying this into Equation \ref{equa1}, we get the following inequality which completes the proof. $$NF(\sigma) \leq \OPT(\sigma) + 2 \OPT(\sigma)  \times \mu  $$ \qed
\end{proof}

\begin{figure}[!t]
\centering
\includegraphics[width=0.75\columnwidth, trim = 0mm 266mm 106mm 0mm, clip]{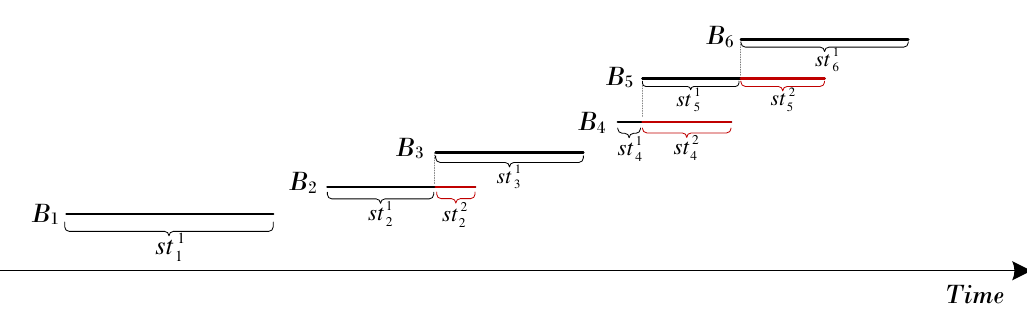}
\caption{The stretch of bins in a packing of the \algo{Next Fit} algorithm. In this example, bins $B_2, B_4,$ and $B_5$ are critical bins. The second periods of bins are highlighted in red. Note that the total lengths of the first periods defines the span of the input sequence. }
\label{fig:intervals}
\end{figure}



\subsection{Improving the Competitive Ratio: \algo{Modified Next Fit} Algorithm}\label{tward}
In this section, we modify the \algo{Next Fit} algorithm to improve its competitive ratio. 
Intuitively speaking, the competitive ratio of \algo{Next Fit} improves for sequences formed by small items. On the other hand, as Proposition \ref{pipipi2} implies, when all items are large, the competitive ratio is independent of $\mu$. This suggest that the competitive ratio might get improved when large and small items are treated separately. A similar approach is used in \cite{LiTang14} to improve the competitive ratio of \algo{First Fit}.\\
%

\begin{myquote}
\emph{\algo{Modified Next Fit:}} The algorithm has a parameter $K\geq 2$ and treats items smaller than $1/K$ using the \algo{Next Fit} strategy. Items larger or equal to $1/K$ are treated separately also using the \algo{Next Fit} strategy. 
\end{myquote}


\begin{theorem}
The competitive ratio of \algo{Modified Next Fit} with parameter $K$ is at most $K \times max\{1,\mu/(K-1)\}+1$.
\end{theorem}

\begin{proof}
Consider a sequence $\sigma$ and let $\sigma^s$ and $\sigma^l$ denote the subsequences of $\sigma$ respectively formed by items smaller and larger or equal to $K$. Recall that the resource utilization of an item is the product of its length and its size, and the total resource utilization of all items in a sequence is a lower bound for the cost of \opt for that sequence.
 As Proposition \ref{pipipi2} suggests, the number of opened bin by Modified Next Fit for items in $\sigma^l$ is no more than 
$k \times util(\sigma^l)$, where $util(\sigma^l)$ is the total utilization of items in $\sigma^l$. 

For placing $\sigma^s$, as the proof of Theorem \ref{NFTH} suggests, the algorithm incurs a cost of at most $\mu \Delta \times  \frac{\omega(\sigma^s)}{1-1/K}   + span(\sigma^s) $, where $\omega(\sigma^s)$ is the total size of item in $\sigma^s$. This will be no more than  $\mu \Delta \times \frac{util(\sigma^s)}{\Delta(1-1/K)}   + span(\sigma^s) $, where $util(\sigma^s)$ is the total resource utilization of $\sigma^s$ (this is because the length of all items is at least $\Delta$).  In total, the cost of the algorithm will be at most $K \times util(\sigma^l) + \mu \times \frac{util(\sigma^s)}{1-1/K}   + span(\sigma^s)$. This is no more than $K \times util(\sigma) \times \max \{1,\frac{\mu}{K-1}\} + span(\sigma)$ where $util(\sigma)$ is the total utilization of items in $\sigma$. 
Since $util(\sigma)$ and $span(\sigma)$ are lower bounds for the cost of \opt, we can conclude the cost of Modified Next Fit is at most $K \times \max\{1,\frac{\mu}{K-1}\} + 1 $. \qed
\end{proof}

When the value of $\mu$ is known to the algorithm, we can define $K$ to be $\mu+1$. In this case, the competitive ratio of \algo{Modified First Fit} will be $\mu+2$.

\begin{proposition}
When the value of $\mu$ is known, there is an online algorithm with competitive ratio of $\mu+2$.

\end{proposition}

\section{Toward Practical Algorithms: \algo{Move To Front} Algorithm}
In the previous sections, we showed that \algo{Next Fit} and a variant of that have promising competitive ratios. These kinds of worst-case guarantees are important in theoretical analysis of the problem. Nevertheless, in practice, beside worst-case guarantees, we are interested in algorithms that also have promising average-case performance. For example, in the case of the classic bin packing problem, \algo{Best Fit} and \algo{First Fit} are preferred over other algorithms in most applications because they have acceptable worst-case performance (although not as good as the \algo{Harmonic} family of algorithms) and superior average-case performance.
We examined different packing algorithms to evaluate their average-case performance. Our experiments are presented in Section \ref{sect:exp} and show that, on average, \algo{Best Fit} has an evident advantage over \algo{First Fit}, and \algo{First Fit} is better than \algo{Next Fit}. These results are in contrast with competitive results (recall that \algo{Best Fit} is not competitive at all) and indicate that the worst-case behaviour and average-case behaviour of these algorithms are quite different.

In this section, we introduce and evaluate the Move-To-Front (\MTF) algorithm for the \prob. We prove that, unlike \algo{Best Fit}, this algorithm is competitive (i.e., provides a worst-case guarantee). Our experiments indicate that \MTF performs better than all other algorithms on randomly generated sequences. \MTF is a simple Any Fit algorithm and runs as fast as BF and FF. Hence, we believe that this is the best algorithm for the \prob. 

\begin{myquote}
\emph{\algo{Move To Front:}} The algorithm maintains a list of open bins. When a new item $x$ arrives, the algorithms checks the bins one by one, starting from the front of the list, until it finds a bin that has enough space for $x$. If no bin has enough space, a new bin is opened for $x$. After placing $x$ in a bin, that bin is moved to the front of the list.
\end{myquote}

\MTF tends to place items which arrive almost at the same time in the same bins. By avoiding placing old an new items in the same bins, the algorithm avoids bad situation in which a bin is open for a small item $x$, and just before $x$ departs another small item $y$ is placed in the bin. 

\begin{theorem}
\algo{Move To Front} has a competitive ratio of at most $6 \mu + 7$.
\end{theorem}

\begin{proof}
Consider the final packing of the algorithm for a sequence $\sigma$. We assume that the sequence is \emph{continuous} in the sense that it has a continuous span, i.e., at each time there is at least one active item and any algorithm maintains at least one open bin. For sequences which are not continuous, at some point, all bins of \MTF and \OPT are closed. In this case, we can divide the sequence into continuous subsequences and apply the same argument for each of them.

We divide the span of the sequence into \emph{periods} of length $(\mu+1) \Deltaa$ (except the last period which might be shorter). For each bin $B$, we define \emph{head}, \emph{tail}, and \emph{body} of $B$ as follows.
If $B$ is opened and closed in the same period, its head is its stretch (interval between its opening and closing) while its body and tail are empty intervals. Otherwise, head of $B$ is the interval between when $B$ and the end of the period in which it is opened. Similarly, tail of $B$ is the interval between the start of the period in which it is closed and when $B$ is closed. Body of $B$ is the interval between its head and tail. Figure \ref{fig:tail} provides an illustration. Let $head(B)$, $body(B)$, and $tail(B)$ indicate the lengths of respectively head, body, and tail of $B$. For the cost that \MTF incurs for $B$ (stretch of $B$) we have: 
$$ stretch(B) = head(B) + body(B) + tail(B) \leq 2 (\mu+1) \Deltaa + body(B)$$
Assume there are $m$ bins opened by \MTF. The algorithm incurs a cost of at most $2\times (\mu+1)$ for head and tail of each bin. We will have:
$$ \MTF(\sigma) \leq 2m (\mu+1) \Deltaa + \sum_{b=1}^m body(B_b)$$
Assume there are $q$ periods in the packing. For each period $P_i$ ($1\leq i \leq q$), let $\alpha(P_i)$ denote the number of bins which have their body in $P_i$, that is, bins which are open at the beginning of the period and remain open till the end. Note that $\alpha(P_i) \geq 1$.
\MTF incurs a cost of $\alpha(P) \times (\mu+1)$ for body of all bins in $P$. We will have: 

\begin{equation}
\MTF(\sigma) \leq 2m (\mu+1) \Deltaa + (\mu+1) \Deltaa \sum_{i=1}^q \alpha(P_i) \label{equik}
\end{equation} 

Next, we consider the cost of \opt for packing $\sigma$. We prove the following claims:\\ \ \\
\emph{Claim 1:} For the number of bins opened by \MTF we have $m \leq 2 \opt(\sigma)/\Deltaa+1$.\\
\emph{Claim 2:} For each period $P$, if $\alpha(P) = 1 $, \OPT incurs a cost of at least $(\mu+1) \Deltaa$.\\
\emph{Claim 3:} For each period $P$, if $\alpha(P) \geq 2$, \OPT incurs a cost of at least $(\alpha(P)-1) \Deltaa /2$.\\

Claim 1 implies that the first term in Equation \ref{equik} is upper bounded by $4(\mu+1) \opt(\sigma)+2(\mu+1)\Deltaa$. Claims 2 implies that in the specified periods, \MTF and \OPT incur the same costs. Claim 3 implies that $\alpha \leq 2 \opt(P)/\Deltaa  + 1$ where $\opt(P)$ is the cost inured by $\opt$ in period $P$. Consequently, the second term in Equation \ref{equik} is upper bounded by 
\begin{align*}
(\mu+1) \Deltaa \sum_{i=1}^q (2 \opt(P_i)/\Deltaa  + 1) & = 2(\mu+1) \sum_{i=1}^q \opt(P_i)  + q (\mu+1) \Deltaa \\ 
 & < 2(\mu+1) \opt(\sigma) + span(\sigma) + (\mu+1) \Deltaa
\end{align*}
%
%
%
The last inequality holds because we have divided the stretch of $\sigma$ into $q$ periods with equal length of $(\mu+1)\Delta$ (except the last period which might be shorter). Adding both terms in Equation \ref{equik}, we get:
\begin{align*} 
\MTF(\sigma) & \leq 4 (\mu+1) \opt(\sigma) + 2 (\mu+1) \Deltaa + 2 (\mu+1) \opt(\sigma) + \opt(\sigma) + (\mu+1) \Deltaa  \\ & =
 (6 \mu+7) \opt(\sigma)+3(\mu+1)\Deltaa
\end{align*}


This implies that the competitive ratio of the algorithm is at most $6 \mu+7$ (note that $3(\mu+1)\Deltaa$ is a constant). 
To complete the proof, it remains to show the above claims hold.

For Claim 1, consider the ordering in which bins are opened and define the amortized weight of a bin (except the last bin in the ordering) as the total size of items in the bin plus the item that causes opening of the next bin. With this definition, every bin (except the last bin) has amortized weight of more than 1. Thus, the total amortized cost of all bins will be more than $m-1$. Each item is counted at most twice in the total amortized cost (once as the item that opens a new bin and once as the member of a bin which cannot fit a new item). Assuming $\omega(\sigma)$ is the total size of items, we will have $m-1 < 2\omega(\sigma)$. 
Note that $\omega(\sigma) \leq util(\sigma)/\Delta \leq \opt(\sigma)/\Delta$. Hence, we get $m \leq 2\opt(\sigma)/\Delta + 1$.

For Claim 2, note that at each time, \opt maintains at least one open bin; otherwise, the sequence is not continuous. 

For Claim 3, let $t$ denote the start time of $P$ and let $B^*$ denote the set of the $\alpha(P)$ bins which have their body stretched along $P$. Consider the time interval $[t+\Delta,t+(\mu+1) \Delta)$. In this interval, any of the bins in $B^*$ receive at least one new item; otherwise, the algorithm would have closed the bin (recall that the length of any item is at most $\mu \Delta$). 
For each bin $B$ in $B^*$, except the last  bin in the list maintained by the algorithm right before time $t+\Delta$, let $t_B$ indicate the time that the bin $B'$ receives an item for the first time (in the interval $[t+\Delta,t+(\mu+1)\Delta)$). Here, $B'$ is the bin that is placed right after $B$ in the mentioned ordering. Define the \emph{critical set} of $B$ as the set of items in $B$ at time $t_B$ plus the item that was placed in $B'$. Note that the total size of items in the critical set of each bin is more than 1.  Hence, the critical items of each bin have a resource utilization of more than $\Delta$ in the interval $[t,t+(\mu+1)\Delta)$ . Since each item belongs to critical sets of at most two bins, the total resource utilization of critical items is at least $(\alpha(P)-1) \Delta/2$ in the interval $[t,t+1+\mu)$. Note that the resource utilization is a lower bound for the cost of \opt in the same interval. \qed

\end{proof}

\begin{figure}[!t]
\centering
\includegraphics[width=0.6\columnwidth, trim = 145mm 0mm 0mm 267mm, clip]{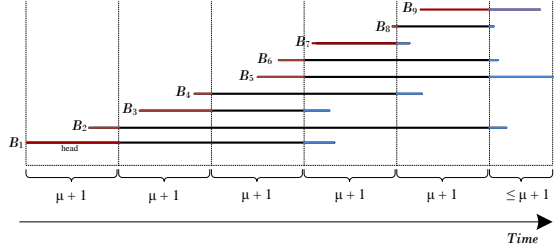}
\caption{The span of an input sequence is divided into periods. The red and blue intervals respectively indicate heads and tails of bins while black intervals are the bodies of the bins.}
\label{fig:tail}
\end{figure}

In fact, the above proof can be extended to any algorithm that maintains a list of bin and places an incoming item in the first bin which has enough space. Such an algorithm might update the list after \emph{placing} items (as \MTF does). In particular, the above analysis also applies for the \algo{First Fit} algorithm. For the \algo{Best Fit} algorithm, the above analysis fails because the order of bins changes when items \emph{depart}. Recall that \algo{Best Fit} is not competitive at all. The bad sequences that cause unbounded competitive ratio of \algo{Best Fit} cause the algorithm keep multiple bins open for an arbitrary long time with an arbitrary small level of $\epsilon$. This cannot happen for \MTF and \algo{First Fit}.

\section{Average-case Analysis: An Experimental Study}\label{sect:exp}
In this section, we study the performance of \MTF compared to other algorithms for the \prob on randomly-generated sequences. We discretize the problem by assuming that the size of bins is an integer $E$ and items have integer sizes in the range $[1,E]$. Moreover, we assume items arrive in discrete time-steps in the range $[1,T-\mu]$ and their length is in the interval $[1,\mu]$. Here, $T$ denotes the span of generated sequences. 
We examine different values of $\mu$ and $T$ for sequences of fixed length. This way, $T$ is a measure of \emph{sparsness} and defines the rate at whcih the items arrive.
Table \ref{tt1} gives details of the datasets that we generated for our experiments. In all cases, both size and length of items are randomly and independently taken from the indicated ranges (assuming a uniform distribution). For each setting of the problem, we run different algorithms on $10^3$ randomly generated sequences. For each sequence, we compute the resource utilization of the sequence as a lower bound for the cost of \opt (see Proposition \ref{prp1}). We use the ratio between the cost of an algorithm and the resource utilization as a measure of performance. 

\begin{table}[!b]
\begin{center}
\scalebox{.9}{
\begin{tabular}{| c | c | c | c |} 
\hline
 Parameter & Description & Value  & Note \\ \hline
 \hline
  $n$ & length of sequences & $10^5$ &  Number of items to be packed \\ \hline
  $\mu$ & maximum length of items & 1,2,5,10,100 & Lengths are picked from the range $[1,\mu]$ \\ \hline
  $T$ & span of sequence & $10^3, 10^4, 10^5$ & Arrival times are picked from the range $[1,T-\mu]$ \\ \hline
  $E$ & bin capacity & $10^3$ & Sizes are picked from the range $[1,E]$\\ \hline
\end{tabular}}
  \caption{A summary of the experimental settings.}\label{tt1}
  \end{center}
\end{table}

The algorithms that we considered in the experiments are \algo{Next Fit}, \algo{Modified Next Fit}, \algo{First Fit}, \algo{Modified First Fit}, \algo{Harmonic}, \algo{Best Fit}, and \algo{Move To Front}. We define the parameters of \algo{Modified Next Fit} and \algo{Modified First Fit} to be respectively $E/(\mu+1)$ and $E/(\mu+7)$. These values ensure that these algorithms achieve their best possible competitive ratio (see Section \ref{tward} and \cite{LiTang14}). Note that the value of $\mu$ is not known to the online algorithm and these algorithms are semi-online in this sense. 
We also consider the \algo{Harmonic} algorithm which classify items by their sizes (using harmonic intervals) into $K$ classes and applies the \algo{Next Fit} strategy for placing items of each class; we assume $K=10$ in our experiments. A straightforward analysis shows that the competitive ratio of the \algo{Harmonic} algorithm is as good as \algo{Modified Next Fit} (with the same parameter $K$). However, similar to the bin packing problem, for the \prob, \algo{Harmonic} seems to have a poor average-case performance. 

Figure \ref{fig:res} shows the average-case performance ratio among all sequences for different algorithms. In most cases, \algo{Move To Front} is the best algorithm. Intuitively, there are two factors which define the quality of a packing. One is how well items are \emph{aligned} to each other. Informally speaking, a packing is well-aligned if items that arrive at the same time are placed together; this ensures that, on expectation, all items of a bin depart also at (almost) the same time. Thus, there is a more chance of saving cost through closing bins. 
 The second factor in defining the quality of a bin is how well the items are packed together according their sizes. Clearly, if items are tightly packed together, there is a save in cost by avoiding opening new bins.  

\begin{figure}[!b]
\begin{center}
\includegraphics[width=\columnwidth]{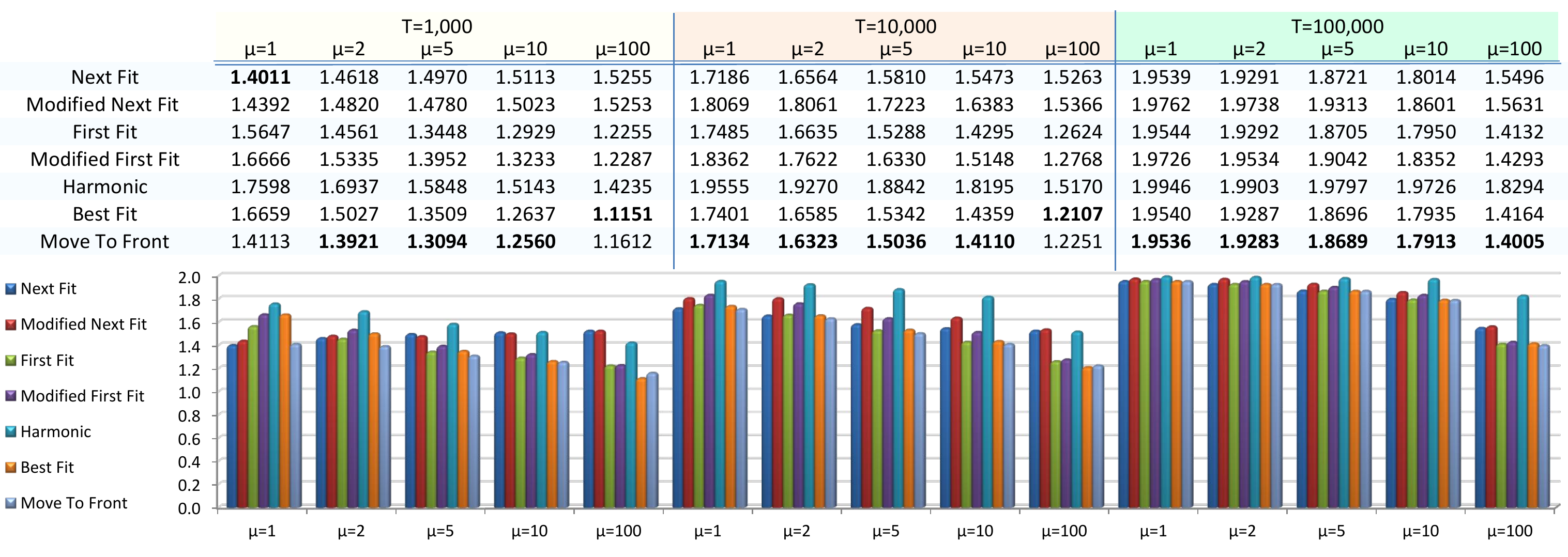}
\caption{Average-case performance ratio of major \prob algorithms, assuming a uniform distribution for the length and size of items. The bold numbers indicate the best algorithm for different values of $T$ and $\mu$. In all cases \MTF is the best or second to the best algorithm. To make comparison easier, the numbers are plotted into a bar diagram.}
\label{fig:res}
\end{center}
\end{figure}

By definition, \algo{Next Fit} results in well-aligned packings; however, it does not packs items as tightly as Any Fit algorithms do. On the other side, \algo{Best Fit} results tight packings which are not necessarily well-aligned. \algo{Move To Front} provides a compromise. 
The packing of \MTF are well-aligned because items placed in the most recent bin have almost same arrival time; meanwhile avoiding to place items in other bins give the a chance of being closed (and saving cost thoroughly). At the same time, as an Any Fit algorithm, \MTF places items almost tightly and does not open a large number of bins as \algo{Next Fit} does. For smaller values of $\mu$, it is more important to achieve well-aligned packings. this is because, when items have roughly same length, there is more benefit in placing them together according to their arrival time since it ensures that they depart at almost same time and their hosting bin gets released. This is particularly more evident for sequences with small span ($T=1,000$). In these sequences, many items appear at the same time and almost all algorithms result in relatively good packing (regarding item sizes).  As a results, for smaller values of $\mu$ and $T$, \algo{Next Fit} performs better relative to other algorithms. In particular, when $\mu=1$ and $T=1,000$, it slightly outperforms \MTF. 
For larger values of $\mu$, it is more important to avoid opening new bins; this is because each bin remains open for a relativity long period of time and one should avoid opening a new bin as much as possible. As a result, when $\mu$ is large ($\mu = 100$), \algo{Best Fit} outperforms \algo{Move To Front}. 

\section{Concluding Remarks}
In this paper, we showed that the \algo{Next Fit} algorithm provides promising worst case guarantees for the \prob. We expect that the same holds for other bounded-space algorithm, e.g., \algo{Harmonic} or \algo{BBF$_2$} of \cite{CsiJoh91}. Unfortunately, these algorithms do not have good average-case performance. To address this issue, we introduced the \MTF algorithm which is a simple and fast Any Fit algorithm which can be regarded as an alternative to the \algo{Best Fit} and \algo{First Fit} algorithms. 
Our experiments indicate that \MTF outperforms other algorithms on the average case for placing sequences that are generated randomly. The closest counterpart of \MTF (regarding average case performance) is the \algo{Best Fit} algorithm which is not good in the worst case as it is not competitive at all. In contrast to \algo{Best Fit}, we proved that \MTF is competitive and has a competitive ratio of at most $6 \mu +7$. We believe this upper bound is not tight and the competitive ratio of \MTF can be improved to $2\mu + 1$ of \NF; we leave this as a future work. Another promising direction for future work is to to provide theoretical upper bounds for the average-case performance of \MTF on sequences that follow arbitrary distributions.

\bibliographystyle{splncs03}
\bibliography{refs2,online,confshort}

\begin{thebibliography}{10}
\providecommand{\url}[1]{\texttt{#1}}
\providecommand{\urlprefix}{URL }

\bibitem{AmazonEC2}
Amazon {EC2}. http://aws.amazon.com/ec2/, accessed: 2014-08-14

\bibitem{Gaikai}
Gaikai. http://www.gaikai.com/, accessed: 2014-08-14

\bibitem{onlive}
{O}n{L}ive. http://www.onlive.com/, accessed: 2014-08-14

\bibitem{BalBek12}
Balogh, J., B\'{e}k\'{e}si, J., Galambos, G.: New lower bounds for certain
  classes of bin packing algorithms. Theoret. Comput. Sci.  440--441,  1--13
  (2012)

\bibitem{ChLaWo08}
Chan, W.T., Lam, T.W., Wong, P.W.H.: Dynamic bin packing of unit fractions
  items. Theoret. Comput. Sci.  409(3),  521--529 (2008)

\bibitem{CoGaJo83}
Coffman, E.G., Garey, M.R., Johnson, D.S.: Dynamic bin packing. SIAM J. Comput.
   12,  227--258 (1983)

\bibitem{CsiJoh91}
Csirik, J., Johnson, D.S.: Bounded space on-line bin packing: {Best} is better
  than first. pp. 309--319 (1991)

\bibitem{Huang13}
Huang, C.Y., Hsu, C.H., Chang, Y.C., Chen, K.T.: Gaminganywhere: An open cloud
  gaming system. In: Proc. the 4th ACM Multimedia Systems Conference. pp.
  36--47. MMSys '13 (2013)

\bibitem{JoDUGG74}
Johnson, D.S., Demers, A., Ullman, J.D., Garey, M.R., Graham, R.L.: Worst-case
  performance bounds for simple one-dimensional packing algorithms. SIAM J.
  Comput.  3,  256--278 (1974)

\bibitem{LiTang14}
Li, Y., Tang, X., Cai, W.: On dynamic bin packing for resource allocation in
  the cloud. pp. 2--11 (2014)

\bibitem{Seid02}
Seiden, S.S.: On the online bin packing problem. J. ACM  49,  640--671 (2002)

\end{thebibliography}



\end{document}